\documentclass[11pt,a4paper]{article}
\usepackage{amsmath}
\usepackage{amsfonts}
\usepackage{amssymb}
\usepackage{color}
\usepackage[utf8]{inputenc}
\usepackage{authblk}
\newtheorem{theorem}{Theorem}[section]
\newtheorem{definition}[theorem]{Definition}
\newtheorem{proposition}[theorem]{Proposition}

\newenvironment{proof}{\paragraph{Proof:}}{\hfill$\square$}
\newcommand{\N}{\mathbb{N}}

\title{On the Algebraic Representation of One-Tape Deterministic Turing Machine}

\author{Yue Liu \thanks{liuyue@fzu.edu.cn}}

\affil{College of Mathematics and Computer Science, Fuzhou University, Fujian, China, 350116}

\date{}
\begin{document}
\maketitle

\begin{abstract}
An algebraic representation of the Turing machines is given, where the configurations of Turing machines are represented by 4 order tensors, and the transition functions by 8 order tensors. Two types of tensor product are defined, one is to model the evolution of the Turing machines, and the other is to model the compositions of transition functions. It is shown that the two types of tensor product are harmonic in the sense that the associate law is obeyed.  
\end{abstract}

\section{Introduction}

Turing machine was invented by Alan Turing in 1936 in \cite{Turing}. And it is still widely used as a fundamental computation model in most of the monographs on theoretical computer science and complexity theory, for example, \cite{Arora}, \cite{Du}, \cite{Garey}, \cite{Hopcroft}, \cite{Martin}, \cite{Sipser}.

Though very simple, Turing machine is believed to have the capability of simulating \emph{every physically realizable computation device} in Church-Turing thesis. A better understanding of Turing machines may lead to better understanding of the general computation procedures and the complexity classes such as P and NP. 

History has witnessed the power of algebraic methods in the studies of geometry, theoretical physics, etc. In this paper, an algebraic representation of the Turing machines is given, where the configurations of Turing machines are represented by 4 order tensors, and the transition functions by 8 order tensors. Two types of tensor product are defined, one is to model the evolution of the Turing machines, and the other is to model the compositions of transition functions. It is shown that the two types of tensor product are harmonic in the sense that the associate law is obeyed. 

\section{The representation} 
\label{sec:the_representation}
 
We adopt the definitions of Turing machines given in \cite{Sipser} in this paper.  A one-tape deterministic Turing machine $M$ is specified by $M=\langle Q,\Gamma,b,\Sigma,\delta,q_s,F \rangle$, where $Q=\{q_1(=q_s),q_2,\ldots,q_n\}$ is the set of states, $\Gamma =\{s_0,s_1,\ldots,s_m\}$ is the alphabet set including $b=s_0$, and $\delta: \Gamma\times Q \to \Gamma\times Q\times \{-1,1\}$ is the transition function. Let the cells on the tape be indexed by the integers $1,2,3,\ldots,$ according to their positions on the tape.

Given an instance of $M$, let $C_t$ be the configuration of the $t$-th step of the machine. Then the configuration $C_t$ can be represented by the 0-1 {\em characteristic tensor} $A(C_t)=(a_{ij,kl})$ of order 4, where $a_{ij,kl}=1$ if and only if the $i$-th cell contains the symbol $s_j$, the state of the controller is $q_k$ and the head is directing to the $l$-th cell; otherwise $a_{ij,kl}=0$.

In the above representation, the configuration $C_t$ contains 4 types of information, the index of a cell,  the symbol written in a cell, the state of the controller, and the position of the head. We may call the information of the index of a cell and  the symbol written in a cell the {\em local information}, since they are owned just by a cell; and the latter 2 types of information the {\em global information}, since they are shared by all the cells.

For technical reasons which will be explained later, we introduce a special state $q_0$. Write $ Q^*=Q\cup\{q_0\} $, extend the definition of $\delta$ to  
$\Gamma\times Q^* \to \Gamma\times Q^*\times \{-1,0,1\}$ as 
\[
	\delta(s,q)=\begin{cases}
		\delta(s,q), &\text{ if }q\in Q\\
		(s,q_0,0),   &\text{ if }q=q_0

	\end{cases}
\]
Write $\delta=(\delta_1,\delta_2,\delta_3) $, where $\delta_1: \Gamma\times Q^* \to \Gamma $, $\delta_2: \Gamma\times Q^* \to Q^*$, and $\delta_3: \Gamma\times Q \to \{-1,0,1\}$.

Write $[k]=\{0,1,2\ldots,k\}. $

Define the {\em characteristic tensor} of $M$ as $B=B(M)=(b^{i_1j_1,k_1l_1}_{i_2j_2,k_2l_2 }) $, where $i_1,i_2\in \N^+$, $j_1,j_2\in [m]$, $k_1,k_2\in [n]$, $l_1,l_2\in\N^+ $, and $b^{i_1j_1,k_1l_1}_{i_2j_2,k_2l_2 }=1$ if one of the following conditions holds
\begin{enumerate}
	\item $k_1\ne 0$, $i_1\ne l_1$, $i_2=i_1$, $j_2=j_1$, $k_2=0$, $l_2=l_1$;
	\item $k_1\ne 0$, $i_1 = l_1$, $i_2=i_1$, $j_2= \delta_1(j_1, k_1)$, $k_2=\delta_2(j_1,k_1)$, $l_2=l_1+\delta_3(j_1,k_1)$.
\end{enumerate}
Otherwise, $b^{i_1j_1,k_1l_1}_{i_2j_2,k_2l_2 }=0$.

A cell is called an {\em active} one if the head is directing to it, otherwise it is called {\em inactive}. The first condition is corresponding to the inactive cells, and the second is corresponding to the active ones. 

In the above representation $B=(b^{i_1j_1,k_1l_1}_{i_2j_2,k_2l_2 }) $, the indexes $i_1j_1,k_1l_1$ reflect the current configuration, and $i_2j_2,k_2l_2 $ reflect the forthcoming one. If $i_1=l_1$, which means the current cell is an active cell, then the forthcoming configuration of this cell is completely determined. Otherwise for the inactive case $i_1\ne l_1$, the cell index and the symbol written (local information) can be determined, but the global information can not be determined alone by the current information, and the next state is set to be the artificially introduced state $q_0$. When the state of a cell is currently $q_0$, by the condition that $b^{i_1j_1,k_1l_1}_{i_2j_2,k_2l_2 }=1$ requires $k_1\ne 0$, this cell will vanish in the forthcoming step, which make sure that the cells with the newly introduced state $q_0$ can not interface the computation procedure furthermore.

The state $q_0$ can also be regarded as a new halting state introduced to deal with the ``illegal'' configurations.

Next we will define a binary operations of tensors with special orders to represent the evolution of the Turing machine. The Einstein notation will be used, which means, for example, that \[
	a_{i_1j_1,k_1l_1} b^{i_1j_1,k_1l_1}_{i_2j_2,k_2l_2}\triangleq \sum_{i_1j_1,k_1l_1} a_{i_1j_1,k_1l_1} b^{i_1j_1,k_1l_1}_{i_2j_2,k_2l_2}
\]
It is always assumed that the summations are well defined and obey the commutative as well as the  distributive law.

\begin{definition}\label{def:Type1prod}
Let $A=(a_{ij,kl})$ be a tensor of order 4, $B=(b^{i_1j_1,k_1l_1}_{i_2j_2,k_2l_2 }) $ be a tensor of order 8. Define the Type I tensor product of $A$ and $B$ to be $C=A\times_1 B$, where $C=(c_{ij,kl})$ is an order 4 tensor such that \[
c_{ij,kl}=\left(\sum_{k_2,l_2}a_{i_1j_1,k_1l_1} b^{i_1j_1,k_1l_1}_{ij,k_2l_2}  \right)\left(\sum_{i_2,j_2}a_{i_1j_1,k_1l_1}b^{i_1j_1,k_1l_1}_{i_2j_2,kl} \right)
\]
\end{definition}

In Defintion \ref{def:Type1prod}, the first summation is to update the local information, and the second is for the global information.

Let $A=(a_{ij,kl})$ be a 4 order tensor such that $k\in [n]$, then $A|_{k\ne 0}$ is the tensor obtained from $A$ by deleting the entries whose $k$ index is 0.

\begin{theorem}\label{thm:product_and_DTM}
	Let $M$ be a DTM, $C_t$ be the $t$-th configuration of $M$ with a given instance, and $B$ be the characteristic tensor of $M$. Suppose that $A_1$ is a 4 order tensor such that $A_1|_{k\ne 0}=A(C_1)$, define $A_{t+1}=A_t\times_1 B $ for $t=1,2,\ldots.$ Then $A(C_t)=A_t|_{k\ne 0} $ .
\end{theorem}
\begin{proof}
	By induction, we only need to prove that if $A_t|_{k\ne 0}=A(C_t) $, then $A(C_{t+1})=A_{t+1}|_{k\ne 0} $.

	Write $A_t=(a_{ij,kl}^{(t)})$ for $t\in \N^+ $ . 

	For $\forall k\ne 0$, we only need to show that $a^{(t+1)}_{ij,kl}\ne 0 $ if and only if the both summations are 1.

	Suppose that $\sum_{k_2,l_2}a^{(t)}_{i_1j_1,k_1l_1}b^{i_1j_1,k_1l_1}_{ij,k_2l_2} \ne 0$.  By the definition of $B$, $b^{i_1j_1,k_1l_1}_{ij,k_2l_2}\ne 0$ indicates that $i_1=i$ and $k_1\ne0$. Since $A_t|_{k\ne 0}$ is the characteristic tensor of $C_t$, there exists only one $a^{(t)}_{i_1j_1,k_1l_1}=1$  ($i_1=i$, $k_1\ne0$). So $\sum_{k_2,l_2}a^{(t)}_{i_1j_1,k_1l_1}b^{i_1j_1,k_1l_1}_{ij,k_2l_2} =1$. It is also easy to check that the nonzero term indicates that the $i$-th cell contains the symbol $s_j$ in $C_{t+1}$. 

	By a similar way, it is easy to check that $\sum_{i_2,j_2}a^{(t)}_{i_1j_1,k_1l_1}b^{i_1j_1,k_1l_1}_{i_2j_2,kl}=1$ if and only if the state in $C_{t+1}$ is $q_k$ and the position of the head is in the $l$-th cell, completing the proof.
\end{proof}

\section{The representation of the composition of a transition function} 
\label{sec:the_representation_of_the_composite_of_the_transition_function}

The evolution of a deterministic Turing machine is actually a Markov process. One of the most successful way to study Markov processes is to use the transition matrices. In the former section, we model the configurations of  deterministic Turing machines with 4 order tensors, the transition functions with 8 order tensors, and the one step evolution is modeled by the Type I tensor production. In this section  the definition of Type II tensor products is given to model the composition of the transition functions. And the Type I tensor product is generalized with respect to that the order of the second tensor can be other than 8.
\begin{definition}
Let $B=(b^{i_1j_1,k_1l_1;\ldots ;i_pj_p,k_pl_p}_{i_{p+1}j_{p+1},k_{p+1}l_{p+1}})$ be a tensor of order $4(p+1)$, $C=(c^{i_1j_1,k_1l_1;\ldots ;i_qj_q,k_ql_q}_{i_{q+1}j_{q+1},k_{q+1}l_{q+1}})$ be a tensor of order $4(q+1) $, $p,q\in \N^+$. Then  the Type II tensor product $B\times_2 C$ is defined to be the tensor $D$ of order $4(2pq+1 )$, where $D=(d^{i_1j_1,k_1l_1;\ldots ;i_{2pq}j_{2pq},k_{2pq}l_{2pq}}_{i_{2pq+1}j_{2pq+1},k_{2pq+1}l_{2pq+1}})$ and
\begin{eqnarray*}
	&   & d^{{i_1j_1,k_1l_1;\ldots ;i_{2pq}j_{2pq},k_{2pq}l_{2pq}}}_{i_{2pq+1}j_{2pq+1},k_{2pq+1}l_{2pq+1}}\\
	& = & \sum_{i_1''j_1'',k_1''l_1'';\ldots;i_q''j_q'',k_q''l_q''}  b^{{i_1j_1,k_1l_1;\ldots ;i_pj_p,k_pl_p}}_{{i_{1}'j_{1}'},k_{1}''l_{1}''}\cdot b^{{i_{p+1}j_{p+1},k_{p+1}l_{p+1};\ldots ;i_{2p}j_{2p},k_{2p}l_{2p}}}_{i_{1}''j_{1}'',{k_{1}'l_{1}'}}\\
	&   & \phantom{ \sum_{i_1''j_1'',k_1''l_1'';\ldots;i_q''j_q'',k_q''l_q''} } \cdot b^{{i_{2p+1}j_{2p+1},k_{2p+1}l_{2p+1};\ldots ;i_{3p}j_{3p},k_{3p}l_{3p}}}_{{i_{2}'j_{2}'},k_{2}''l_{2}''}\cdot b^{{i_{3p+1}j_{3p+1},k_{3p+1}l_{3p+1};\ldots ;i_{4p}j_{4p},k_{4p}l_{4p}}}_{i_{2}''j_{2}'',{k_{2}'l_{2}'}}\\
	&   & \phantom{ \sum_{i_1''j_1'',k_1''l_1'';\ldots;i_q''j_q'',k_q''l_q''} } \cdots \quad   \cdots \cdot b^{{i_{(2q-1)p+1}j_{(2q-1)p+1},k_{(2q-1)p+1}j_{(2q-1)p+1};\ldots;i_{2pq}j_{2pq},k_{2pq}l_{2pq}}}_{i_q''j_q'',{k_q'l_q'}}\\
	&   & \phantom{ \sum_{i_1''j_1'',k_1''l_1'';\ldots;i_q''j_q'',k_q''l_q''} } \cdot c^{{i_1'j_1',k_1'l_1';\ldots;i_q'j_q',k_q'l_q'}}_{{i_{2pq+1}j_{2pq+1},k_{2pq+1}l_{2pq+1}}}.
\end{eqnarray*}
\end{definition}

\begin{definition}
	Let $A=(a_{ij,kl})$, $B=(b^{i_1j_1,k_1l_1;\ldots ;i_pj_p,k_pl_p}_{i_{p+1}j_{p+1},k_{p+1}l_{p+1}})$. Then the Type I tensor product $A\times_1 B$ is defined to be the tensor $C=(c_{ij,kl})$, where 
	\begin{eqnarray*}
		c_{ij,kl} & = & \left(\sum_{k',l'} a_{i_1j_1,k_1l_1}\cdots a_{i_pj_p,k_pl_p}b^{i_1j_1,k_1l_1;\ldots ;i_pj_p,k_pl_p}_ {ij,k'l'}\right)\left(\sum_{i',j'} a_{i_1j_1,k_1l_1}\cdots a_{i_pj_p,k_pl_p}b^{i_1j_1,k_1l_1;\ldots ;i_pj_p,k_pl_p}_ {i'j',kl}\right).
	\end{eqnarray*}
\end{definition}

The following proposition can be verified directly form the definitions.

\begin{proposition}
	Let $A=(a_{ij,kl})$, $B=(b^{i_1j_1,k_1l_1;\ldots ;i_pj_p,k_pl_p}_{i_{p+1}j_{p+1},k_{p+1}l_{p+1}})$ and $C=(c^{i_1j_1,k_1l_1;\ldots ;i_qj_q,k_ql_q}_{i_{q+1}j_{q+1},k_{q+1}l_{q+1}})$. Then \[
		(A\times_1 B) \times_1 C = A\times_1 (B\times_2 C),
	\] and the Type II tensor product $\times_2$ obeys the associative law.
\end{proposition}

{\bf Acknowledgment }

This work was supported by National Natural Science Foundation of China 11571075. Part of the work was done when the author was visiting the Department of Mathematics, the College of William \& Mary. The visiting was supported by China Scholarship Council (CSC).

\end{document}